\let\oldvec\vec
\documentclass[runningheads]{llncs}
\let\vec\oldvec

\usepackage{times}  
\usepackage{helvet}  
\usepackage{courier}  
\usepackage{graphicx}  
\usepackage{amssymb}
\usepackage{amsmath}
\usepackage[group-separator={,}]{siunitx}
\usepackage{multirow}

\newcolumntype{C}{>{\centering\arraybackslash}p{1.5cm}}
\newcommand{\myosd}{\textsc{On\-li\-ne Se\-ri\-al Di\-c\-ta\-to\-r}}
\newcommand{\myorp}{\textsc{On\-li\-ne R\-an\-d\-om Pr\-io\-ri\-ty}}
\newcommand{\mymax}{\textsc{Ma\-x\-im\-um Li\-ke}}
\newcommand{\mypar}{\textsc{Pa\-re\-to Li\-ke}}
\newcommand{\mylike}{\textsc{Li\-ke}}
\newcommand{\myblike}{\textsc{Ba\-la\-n\-ced Li\-ke}}

\begin{document}

\title{Strategy-proofness, Envy-freeness and Pareto efficiency \\ in Online Fair Division with Additive Utilities\thanks{Funded by the European Research Council under the Horizon 2020 Programme via AMPLify 670077.}}
\titlerunning{Online Fair Division with Additive Utilities}
\author{Martin Aleksandrov \and Toby Walsh}
\authorrunning{M. Aleksandrov \and T. Walsh}
\institute{Technical University Berlin, Germany \email{\{martin.aleksandrov,toby.walsh\}@tu-berlin.de}}

\maketitle

\begin{abstract}
We consider fair division problems where indivisible items arrive one by one in an online fashion and are allocated immediately to agents who have additive utilities over these items. Many existing offline mechanisms do not work in this online setting. In addition, many existing axiomatic results often do not transfer from the offline to the online setting. For this reason, we propose here three \emph{new} online mechanisms, as well as consider the axiomatic properties of three previously proposed online mechanisms. In this paper, we use these mechanisms and characterize classes of online mechanisms that are strategy-proof, and return envy-free and Pareto efficient allocations, as well as combinations of these properties. Finally, we identify an important impossibility result.

\end{abstract}

\section{Introduction}\label{sec:intro}

Fair division is an important problem facing our society today as increasing economical, environmental, and other pressures require us to try to do more with limited resources. An especially challenging form of fair division is when we are allocating available resources in an \emph{online} fashion with only partial knowledge of the future resources and agent's preferences for these resources. There are many applications of online fair division for \emph{social good}. For example, when a kidney is donated, it must be allocated to a patient within a few hours. As a second example, food items arrive at a food bank and must be allocated and distributed to charities promptly. As a third example, when allocating charging slots to electric cars, we may not know when or where cars will arrive for charging. As a fourth example, when managing a river, we might start allocating irrigation water to farmers today, not knowing how much it will rain the next month. As a fifth example, when allocating memory to cloud services, we may not know what and how many services are requested in the next moment.

The online nature of such fair division problems changes the mechanisms available to allocate items. For example, with the well-known (offline) \emph{sequential allocation} mechanism, agents pick their most preferred remaining items in turns. In an online setting, an agent's most preferred item may not be currently (or even ever) available. To tackle this, we propose three \emph{new} - \myosd, \myorp\ and \mypar\ - as well as study three existing - \mylike, \myblike\ and \mymax- online mechanisms. The online nature also means we may need to consider \emph{new} axiomatic properties. For example, in deciding if agents have any incentive to misreport preferences in an online setting, we may consider the past fixed but the future unknown. This leads to a \emph{new} and weaker form of \emph{online strategy-proofness} (OSP). Therefore, it might be easier to achieve strategy-proofness in an online than in an offline setting. Also, we give a \emph{new} and stronger form of envy-freeness, called \emph{shared envy-freeness} (SEF), in which agents might be envious of each other but only over the items that they like in common. For example, in the paper assignment problem, reviewers tend to bid for papers in their field of expertise and not for papers outside this field \cite{lian2018}. In this context, SEF guarantees envy-freeness across the different fields.  

We provide characterization results for strategy-proofness (SP), envy-freeness (EF) and Pareto efficiency (PE). For example, we characterize completely the class of online mechanisms that are SP, and the class of online mechanisms that are PE ex post. We also characterize the class of SP and EF mechanisms. Thus, a mechanism for online fair division is SP and EF ex ante iff it returns the same random assignment as \mylike. The same holds for SEF ex ante mechanisms. Also, we prove that a mechanism is SP, PE ex post and EF ex ante iff it returns the same probability distribution of allocations as \myorp. We further give an important impossibility result. In offline fair division, stochastic Pareto efficiency and envy-freeness are always possible simultaneously (e.g.\ the probabilistic serial mechanism \cite{bogomolnaia2001}). However, we prove that no online mechanism can be both Pareto efficient ex ante and envy-free ex ante.

\section{Related Work}\label{sec:rel}

We consider the model of online fair division from \cite{walsh2014} in which items are indivisible and arrive one-by-one over time. We primarily contrast our characterization results with similar results in (offline) fair division. For example, we prove that \emph{no} online mechanism can be both PE and EF ex ante. By comparison, the (offline) probabilistic serial mechanism satisfies both stochastic PE and EF \cite{bogomolnaia2001}. In fact, it follows from our results that there could be an unbounded number of mechanisms that are just PE ex ante or EF ex ante. We can show that other (offline) characterizations (e.g.\ \cite{brams2005,manea2007}) break in the online setting as well. By comparison, as online mechanisms can be applied to offline problems by picking a sequence of the items, our results can be mapped into such settings. For example, our \mypar\ mechanism returns all possible PE ex post allocations in the offline problem. As a result, this mechanism characterizes the set of offline such mechanisms. As another example, we prove that \myorp\ is SP and PE ex post, but not PE ex ante. With this mechanism, agents with the same cardinal utilities receive the same expected utilities (i.e.\ it is symmetric). This is in-line with the impossibility result that \emph{no} (offline or online) mechanism for offline matching is SP, PE ex ante and symmetric \cite{zhou1990}. Yet more related results are shown in many other fair division (e.g.\ \cite{chevaleyre2008,freeman2018,kash2014,walsh2011}), voting (e.g.\ \cite{gibbard1973,xia2010}) and kidney exchange (e.g.\ \cite{dickerson2012,dickerson2015}) settings. Our results can also be mapped to such settings.

\section{Online and Additive Fair Division}\label{sec:model}

An online fair division \emph{instance} consists of a set of \emph{agents} $N=\lbrace 1,\ldots,n\rbrace$, and an ordered set of indivisible \emph{items} $O=\lbrace o_1,\ldots,o_m\rbrace$. We suppose that item $o_j$ arrives at round $j$ when each agent $i\in N$ becomes aware of their sincere \emph{utility} $u_{ij}\in\mathbb{R}_{\geq 0}$ and places a 
possibly strategic \emph{bid} $v_{ij}\in\mathbb{R}_{\geq 0}$ for $o_j$. We suppose at least one agent has positive utility for every item as, otherwise, we can simply discard the item. We use \emph{online} mechanisms that allocate $o_j$ immediately, supposing the allocation of $o_1$ to $o_{j-1}$ is fixed and there is \emph{no information} of $o_{j+1}$ to $o_m$. We consider only \emph{non-wasteful} mechanisms that share the probability of 1 for $o_j$ only among agents that bid positively for it if there is at least one such agent and, otherwise, discard $o_j$. 

An \emph{allocation} $\pi_j$ of $o_1$ to $o_j$ gives a bundle of items $\pi_{ji}$ to each agent $i\in N$ such that $\bigcup_{i \in N} \pi_{ji} =\lbrace o_1,\ldots,o_j\rbrace$ and $\pi_{ji}\cap\pi_{jk}=\emptyset$ for each $i \neq k$. We write $u_{ik}(\pi_j)$ for the \emph{utility} of agent $i\in N$ for $\pi_{jk}$. We write $u_i(\pi_j)$ for $u_{ii}(\pi_j)$. A mechanism induces a probability distribution over the set $\Pi_j$ of all allocations of items $o_1$ to $o_j$. We write $\overline{u}_{ik}(\Pi_j)$ for the \emph{expected utility} of agent $i\in N$ for the expected allocation of agent $k\in N$ and $p_{ik}(\Pi_j)$ for the \emph{probability} of agent $i\in N$ for item $o_k$ in this distribution. We write $\overline{u}_i(\Pi_j)$ for $\overline{u}_{ii}(\Pi_j)$ and $p_i(\Pi_j)$ for $p_{ij}(\Pi_j)$. We suppose \emph{additive} utilities and expected utilities. 

\begin{center}
$u_{ik}(\pi_j)=\displaystyle\sum_{o_h\in\pi_{jk}} u_{ih}$ \hspace{1cm} $\overline{u}_{ik}(\Pi_j)=\displaystyle \sum_{h=1}^j p_{kh}(\Pi_j)\cdot u_{ih}$
\end{center}

We consider three common properties of mechanisms: strategy-proofness, envy-freeness and Pareto efficiency.

\begin{definition} (SP)
A mechanism is \emph{strategy-proof (SP)} if, for each instance with $m\in\mathbb{N}$ items, no agent $i\in N$ can strictly increase $\overline{u}_{i}(\Pi_m)$ by reporting any sequence $v_{i1},\ldots,v_{im}$ other than $u_{i1},\ldots,u_{im}$, supposing all other agents bid sincerely for items $o_1$ to $o_m$.
\end{definition}

\begin{definition} (EF)
A mechanism is \emph{envy-free ex post (EFP)} iff, for each instance with $m\in\mathbb{N}$ items and allocation $\pi_m\in \Pi_m$ returned by the mechanism with positive probability, $\forall i,k\in N:u_{ii}(\pi_m)\geq u_{ik}(\pi_m)$. A mechanism is {\em envy-free ex ante (EFA)} iff, for each instance with $m\in\mathbb{N}$ items, $\forall i,k\in N:\overline{u}_{ii}(\Pi_m)\geq \overline{u}_{ik}(\Pi_m)$.
\end{definition}

\begin{definition} (PE)
A mechanism is \emph{Pareto efficient ex post (PEP)} iff, for each instance with $m\in\mathbb{N}$ items and allocation $\pi_m\in \Pi_m$ returned by the mechanism with positive probability, no $\pi^{\prime}_m\in\Pi_m$ is such that $\forall i\in N:u_i(\pi^{\prime}_m)\geq u_i(\pi_m)$ and $\exists k\in N:u_k(\pi^{\prime}_m)>u_k(\pi_m)$. Also, it is \emph{Pareto efficient ex ante (PEA)} iff, no mechanism gives at least $\overline{u}_{i}(\Pi_m)$ to each $i\in N$ and more than $\overline{u}_{k}(\Pi_m)$ to some $k\in N$.
\end{definition}

To characterize SP, EF and PE mechanisms, we will use two equivalence relations between outcomes of mechanisms.  We say that two mechanisms are \emph{ex ante equivalent} iff, for each instance of $m\in\mathbb{N}$ items, agent $i\in N$ and item $o_j\in O$, the probabilities of $i$ for $o_j$ under both mechanisms are equal, whilst these mechanisms are \emph{ex post equivalent} iff, for each instance of $m\in\mathbb{N}$ items and allocation $\pi_m\in\Pi_m$, the probabilities of $\pi_m$ under both mechanisms are equal (i.e.\  each of the two mechanisms returns an identical distribution of allocations). 

\section{Six Cardinal Mechanisms}\label{sec:mec}

Many offline mechanisms cannot be used in the online setting because only one item is available at any time. For this reason, we propose three \emph{new} as well as study three existing online mechanisms. For every arriving item $o_j$, each mechanism first computes a set of agents feasible for $o_j$ given an allocation $\pi_{j-1}\in\Pi_{j-1}$. An agent that is feasible for $o_j$ then receives it with \emph{conditional probability} that is uniform with respect to the other agents that are feasible for $o_j$. Thus, for the first $j$ items, each mechanism returns a probability distribution over $\Pi_j$ and an actual allocation with some positive probability that is obtained as a product of $j$ conditional randomizations.

\begin{itemize}
\item \myosd: it has a strict priority order $\sigma$ of the agents prior to round one, and the unique feasible agent for $o_j$ is the first agent in $\sigma$ that bids positively for $o_j$. 

\item \myorp: it draws uniformly at random a strict priority order $\sigma$ of the agents prior to round one, and runs \myosd\ with it. 
  
\item \mypar: agent $i\in N$ is feasible for $o_j$ if extending $\pi_{j-1}$ by allocating $o_j$ to $i$ is Pareto efficient ex post. 

\item \mylike: agent $i\in N$ is feasible for $o_j$ if $v_{ij}>0$ \cite{aleksandrov2015ijcai}. 

\item \myblike: agent $i\in N$ is feasible for $o_j$ if $v_{ij}>0$ and $i$ has the fewest items in $\pi_{j-1}$ among those with positive bids for $o_j$ \cite{aleksandrov2015ijcai}.

\item \mymax: agent $i\in N$ is feasible for $o_j$ if $v_{ij}=\max_{k\in N} v_{kj}$ \cite{aleksandrov2017mcm}. 
\end{itemize}

In Example~\ref{exp:one}, we demonstrate that these mechanisms may return distributions of allocations that are different from each other.

\begin{example}\label{exp:one}
Let us consider an instance with $N=\lbrace 1,2\rbrace$ and $O=\lbrace o_1,o_2\rbrace$. The utilities of agents for items are given in the below table.

\begin{center}
\begin{tabular}{|c|cc|} \hline
& item $o_1$ & item $o_2$ \\ \hline
agent 1 & \num{1} & \num{2}  \\
agent 2 & \num{2} & \num{1} \\ \hline
\end{tabular}
\end{center}

In this instance, supposing sincere bidding, there are 4 possible allocations: $\pi^1=(\lbrace o_1,o_2\rbrace,\emptyset)$, $\pi^2=(\emptyset,\lbrace o_1,o_2\rbrace)$, $\pi^3=(\lbrace o_1\rbrace,\lbrace o_2\rbrace)$, and $\pi^4=(\lbrace o_2\rbrace,\lbrace o_1\rbrace)$. \myosd\ with fixed $\sigma=(1,2)$ returns $\pi^1$ with probability $1$, \myorp\ returns $\pi^1$ and $\pi^2$ with probabilities $1/2$, \mypar\ returns $\pi^1$ with probability $1/2$, $\pi^2$ and $\pi^4$ with probabilities $1/4$, \mylike\ returns $\pi^1$ to $\pi^4$ with probabilities $1/4$, \myblike\ returns $\pi^3$ and $\pi^4$ with probabilities $1/2$, and \mymax\ returns $\pi^4$ with probability $1$.\qed
\end{example}

We note that the \myosd\ mechanism is similar to the (offline) \emph{serial dictatorship} mechanism \cite{svensson1999}. However, agents have no quota on the number of items they receive with \myosd, and only take items for which they declare non-zero utility. The \myorp\ mechanism is also similar to the (offline) \emph{random priority} mechanism \cite{abdulkadiroglu1998}. Finally, the \mylike\ mechanism can be seen as the online analog of the (offline) \emph{probabilistic serial} mechanism (see \cite{bogomolnaia2001}) with agents ``eating'' each next item which they like.

\section{Strategy-Proofness}\label{sec:sp}

We begin by considering strategic behavior of agents. We provide a simple characterization of mechanisms that are strategy-proof. For $i\in N$, we say that $p_i(\Pi_j)$ is a \emph{step} function iff it is 0 if $v_{ij}=0$ and it admits the same value for any bid $v_{ij}>0$ supposing the bids of the other agents for $o_1$ to $o_j$, and the bids of agent $i$ for $o_1$ to $o_{j-1}$ are fixed. A mechanism is a \emph{step} mechanism iff, for each instance with $m\in\mathbb{N}$ items, $i\in N$ and $o_j\in O$, $p_i(\Pi_j)$ is a step function. For $i\in N$, we say that $p_i(\Pi_j)$ is a \emph{memoryless} function iff it takes the same value for all possible bids $v_{i1}$ to $v_{i(j-1)}$ of agent $i$ for items $o_1$ to $o_{j-1}$ given fixed bid $v_{ij}$ of agent $i$ for item $o_j$ and fixed bids of the other agents for items $o_1$ to $o_j$. A mechanism is a {\em memoryless} mechanism iff, for each instance with $m\in\mathbb{N}$ items, $i\in N$ and $o_j\in O$, $p_i(\Pi_j)$ is a memoryless function. 

With a step mechanism, $p_i(\Pi_j)$ does not depend on the size of an agent's non-zero bid for item $o_j$ but it may depend on the allocation history. By comparison, with a memoryless mechanism, $p_i(\Pi_j)$ may depend on the size of their non-zero bid for item $o_j$ but not on the allocation history. As a consequence, with a memoryless step mechanism, $p_i(\Pi_j)$ depends only on the combination of the non-zero bids for item $o_j$.

\begin{theorem}\label{thm:one}
A non-wasteful mechanism for online fair division is strategy-proof iff it is
a memoryless step mechanism.
\end{theorem}

\begin{proof}
Pick $i\in N$ in an instance. Let us view $\overline{u}_{i}(\Pi_j)$ and $p_{i}(\Pi_j)$ as functions of $v_{i1}$ to $v_{ij}$. That is, we write $\overline{u}_{i}(\Pi_j)=\overline{u}_{i}(v_{i1},\ldots,v_{ij})$ and $p_i(\Pi_j)=p_i(v_{i1},\ldots,v_{ij})$. Consider a memoryless step mechanism. Suppose now that all agents bid sincerely. Then, $\overline{u}_i(u_{i1},\ldots, u_{im})=\sum_{j=1}^m p_i(u_{i1},\ldots, u_{ij})\cdot u_{ij}$. Suppose next that only $i$ bids strategically $v_{i1}$ to $v_{im}$. Then, $\overline{u}_i(v_{i1},\ldots, v_{im})=\sum_{j=1}^m p_i(v_{i1},\ldots, v_{ij})\cdot u_{ij}$. For each $o_j$ with $v_{ij}=u_{ij}$, $p_i(v_{i1},\ldots, v_{ij})\cdot u_{ij}=p_i(u_{i1},\ldots, u_{ij})\cdot u_{ij}$ as the mechanism is a memoryless step. For each $o_j$ with $v_{ij}>0$ and $u_{ij}=0$, $p_i(v_{i1},\ldots, v_{ij})\cdot u_{ij}=p_i(u_{i1},\ldots, u_{ij})\cdot u_{ij}=0$. For each $o_j$ with $v_{ij}=0$ and $u_{ij}>0$, $p_i(v_{i1},\ldots, v_{ij})\cdot u_{ij}=0$ and $p_i(u_{i1},\ldots, u_{ij})\cdot u_{ij}\geq 0$ as the mechanism is non-wasteful. Consequently, the mechanism is strategy-proof.

Consider a strategy-proof mechanism. First, assume that it is not a step and $p_i(u_{i1},$ $\ldots,u_{i(j-1)},v_{ij})$ admits different values for different positive values of $v_{ij}$ supposing that the bids of other agents for items $o_1$ to $o_j$ are fixed. WLOG, we can suppose that item $o_j$ is the last item to arrive. We can also suppose $u_{ij}>0$ as the case $u_{ij}=0$ is trivial. Agent $i$ has an incentive to report $v_{ij}>u_{ij}$ (or $v_{ij}<u_{ij}$) and, thus, strictly increase $p_i(u_{i1},\ldots,u_{i(j-1)},u_{ij})$ and $\overline{u}_{i}(u_{i1},\ldots,u_{i(j-1)},u_{ij})$. Second, assume that the mechanism is a step but not memoryless. Suppose that agent $i$ gets different probabilities for item $o_j$ for alternative bids $v_{ik}$ compared to their sincere bids $u_{ik}$ with $k<j$. WLOG, for each $o_k$ with $k<j$, we suppose that $p_i(v_{i1},\ldots, v_{ik})=p_i(u_{i1},\ldots, u_{ik})$. Otherwise, we truncate the problem to the first such round $j$. WLOG, we also suppose that $p_i(v_{i1},\ldots, v_{i(j-1)},u_{ij})>p_i(u_{i1},\ldots, u_{i(j-1)},u_{ij})$. Otherwise, we swap $v_{ik}$ for $u_{ik}$ for $k<j$. We let agent $i$ have utility $1$ for all items except $o_j$ and utility $j$ for $o_j$. Thus, the bids $v_{ik}$ increase the expected utility of agent $i$ compared to the bids $u_{ik}$. We reached contradictions under both assumptions. \qed
\end{proof}

The \mylike\ mechanism is a memoryless step and so is strategy-proof. We observe that the \myosd\ and \myorp\ mechanisms are also memoryless steps and, hence, are also both strategy-proof. On the other hand, the \myblike\ mechanism is just a step mechanism and is neither memoryless nor strategy-proof. Furthermore, the \mymax\ mechanism is only memoryless and the \mypar\ mechanism is neither a step nor a memoryless mechanism. Consequently, these two mechanisms are not strategy-proof. 

Thus far, we have made the strong assumption that an agent has complete knowledge of any future items. In practice, agents may have limited or even no knowledge about the future. We next capture this formally in terms of a definition of a weaker form of strategy-proofness. 

\begin{definition} (OSP)
A mechanism is \emph{online strategy-proof (OSP)} if, for each instance with $m\in\mathbb{N}$ items and $j\in\lbrace 1,\ldots,m\rbrace$, no agent $i\in N$ can strictly increase $\overline{u}_{i}(\Pi_j)$ by reporting any bid $v_{ij}$ other than $u_{ij}$, supposing agent $i$ bids sincerely for $o_1$ to $o_{j-1}$ and all other agents bid sincerely for items $o_1$ to $o_j$.
\end{definition}

Indeed, it is harder for an agent to benefit from a strategic bidding with only partial information of the future. For this reason, many mechanisms that are not strategy-proof are online strategy-proof. For example, the \myblike\ mechanism is online strategy-proof with no knowledge of future items, but stops being strategy-proof with complete knowledge of these future items even if all utilities are just 0 or 1  \cite{aleksandrov2015ijcai}. In the other direction, it is easy to show that a mechanism that is strategy-proof is also online strategy-proof. The reason for this is simple. If an agent cannot increase their expected utility by misreporting their utilities for any subset of items, then they cannot do it by misreporting their utility for any individual item, including the last one. We give a simple characterization of mechanisms that are online strategy-proof.

\begin{theorem}\label{thm:two}
A non-wasteful mechanism for online fair division is online strategy-proof iff it is a step mechanism. 
\end{theorem}

\begin{proof}
We show the ``if'' direction. Suppose the mechanism is a step. Consider an instance, an agent $i\in N$ and an item $o_j$. The allocation of this item does not have an impact on the allocation of earlier items as this is now fixed. If $u_{ij}>0$, then agent $i$ has no incentive to report $0$ for it as their expected utility can only decrease, and also has no incentive to report any positive value $v_{ij}\not=u_{ij}$ as their probability for item $o_j$ is a step function. If $u_{ij}=0$, then agent $i$ has no incentive to report $v_{ij}>0$ as their expected utility cannot increase. Hence, $i$ cannot increase $\overline{u}_{i}(\Pi_j)$. The mechanism is online strategy-proof. We next sketch the ``only if'' direction. Suppose the mechanism is not a step. The result follows by the second part of the proof of Theorem~\ref{thm:one}.\qed
\end{proof}

It follows immediately that the \myosd, \myorp, \mylike\ and \myblike\ mechanisms are all online strategy-proof. In contrast, the \mymax\ and \mypar\ mechanisms are not as they are not steps and agents have an incentive to report a larger bid for an item.

To sum up, we might use the \myosd, \myorp, or \mylike\ mechanism for strategy-proofness with complete information. However, for online strategy-proofness with no information about future items, we can also use the \myblike\ mechanism.

\section{Envy-Freeness}\label{sec:ef}

We continue with envy-freeness. We suppose agents bid sincerely. This might be because we use a mechanism that is strategy-proof or online strategy-proof. There is \emph{no} envy-free ex post mechanism \cite{aleksandrov2015ijcai}. We, therefore, mainly focus on fairness in expectation. Uncertainty about the future means that envy-freeness ex ante is now harder to achieve than in the offline setting. Nevertheless, it is always \emph{possible} as the \mylike\ mechanism is envy-free ex ante. 

By Example~\ref{exp:one}, the \myorp\ and \mylike\ mechanisms can return different ex post allocations. Nevertheless, they are ex ante equivalent and, therefore, envy-free ex ante. Unfortunately, ex ante equivalence to the \mylike\ mechanism only provides a partial characterization as there is an unbounded number of envy-free ex ante mechanisms that are \emph{not} ex ante equivalent to it. We show this in Example~\ref{exp:two}.

\begin{example}\label{exp:two}
Let us consider the fair division of items $o_1$ and $o_2$ to agents $1$ and $2$ with utilities as follows: $u_{11}=1$, $u_{12}=1$, $u_{21}=0$ and $u_{22}=1$. Further, consider the mechanism that works as \mylike\ on each instance except on this one in which it gives item $o_2$ to agent $2$ with some probability in $(1/2,1]$. This mechanism is envy-free ex ante but it is not ex ante equivalent to \mylike.\qed
\end{example}

In Example~\ref{exp:two}, the mechanism is neither memoryless, nor a step. Therefore, by Theorem~\ref{thm:one}, it is not strategy-proof. However, we can give a complete characterization of \emph{all} strategy-proof and envy-free ex ante mechanisms.

\begin{theorem}\label{thm:three}
A non-wasteful mechanism for online fair division is strategy-proof and envy-free ex ante iff it is ex ante equivalent to the \mylike\ mechanism. 
\end{theorem}

\begin{proof}
If a mechanism is ex ante equivalent to \mylike, then it is envy-free ex ante and a memoryless step by the definition of \mylike. By Theorem~\ref{thm:two}, the mechanism is strategy-proof. If a mechanism is envy-free ex ante and strategy-proof, then it is a memoryless step. We show that it is ex ante equivalent to \mylike\ by induction on the round number $j$. In the base case, the mechanism is clearly ex ante equivalent to \mylike. In the step case, suppose that the mechanism is ex ante equivalent to \mylike\ for items $o_1$ to $o_{j-1}$ (i.e.\ hypothesis) but not for item $o_j$. That is, there are two agents $i,k\in N$ that like item $o_j$ with $p_{i}(\Pi_j)<p_{k}(\Pi_j)$. As the mechanism is envy-free ex ante up to round $(j-1)$, we have that $\overline{u}_{ii}(\Pi_{j-1})\geq \overline{u}_{ik}(\Pi_{j-1})$. As the mechanism is memoryless step, we can suppose that $u_{ij}=1-(\overline{u}_{ik}(\Pi_{j-1})-\overline{u}_{ii}(\Pi_{j-1}))/(p_{k}(\Pi_j)-p_{i}(\Pi_j))>0$. We, hence, obtain that $\overline{u}_{ik}(\Pi_{j-1})-\overline{u}_{ii}(\Pi_{j-1})+(p_{k}(\Pi_j)-p_{i}(\Pi_j))\cdot u_{ij}>0$, or $i$ envies ex ante $k$ for $o_1$ to $o_j$. This contradicts the fact that the mechanism is envy-free ex ante up to round $j$. Consequently, $p_{i}(\Pi_j)=p_{k}(\Pi_j)$. The result follows.
\qed
\end{proof}

We can give similar results if we weaken strategy-proof mechanisms to memoryless or step mechanisms. We omit these proofs for reasons of space.

\begin{proposition}\label{pro:one}
A step mechanism for online fair division is envy-free ex ante iff it is ex ante equivalent to the \mylike\ mechanism. 
\end{proposition}

\begin{proposition}\label{pro:two}
A memoryless mechanism for online fair division is envy-free ex ante iff it is ex ante equivalent to the \mylike\ mechanism. 
\end{proposition}

On a restricted preference domain, the \mylike\ mechanism characterizes all envy-free ex ante mechanisms, even without the assumption of strategy-proofness. The following result applies to common domains of positive cardinal, identical cardinal, identical ordinal, Borda (e.g.\ $1,2,\ldots,m$) or lexicographic (e.g.\ $2^0,2^1,\ldots,2^m$) utilities. This result holds for \emph{wasteful} (i.e.\ not non-wasteful) mechanisms as well.

\begin{theorem}\label{thm:four}
With non-zero cardinal utilities, a mechanism for online fair division is envy-free ex ante iff it is ex ante equivalent to the \mylike\ mechanism.
\end{theorem}

\begin{proof}
We first show the ``if'' direction. If a mechanism is ex ante equivalent to \mylike, then it is envy-free ex ante as \mylike. We next show the ``only if'' direction. The proof is by induction as in Theorem~\ref{thm:three}. In the step case, we consider $i,k\in N$ that like $o_j$. We have that $\overline{u}_{ii}(\Pi_{j-1})=\overline{u}_{ik}(\Pi_{j-1})$ and $\overline{u}_{kk}(\Pi_{j-1})=\overline{u}_{ki}(\Pi_{j-1})$ as the cardinal utilities are non-zero and the mechanism is ex ante equivalent to \mylike\ for $o_1$ to $o_{j-1}$ by the hypothesis. Hence, $p_{i}(\Pi_j)=p_{k}(\Pi_j)$ as the mechanism is envy-free ex ante up to round $j$. \qed
\end{proof}

We can also completely characterize a stronger notion of envy-freeness even with general utilities. Shared envy-freeness requires that each pair of agents are envy-free of each other only over the items that both agents in the pair like in common. We write $u^{\mbox{\scriptsize SEFP}}_{ik}(\pi_j)$ for the utility of agent $i\in N$ over the items in $\pi_{ji}$ that both agents $i$ and $k\in N$ like. We write $\overline{u}^{\mbox{\scriptsize SEFA}}_{ik}(\Pi_{j})$ for the expected utility of agent $i\in N$ over the items $o_1$ to $o_j$ that both agents $i$ and $k\in N$ like. 

\begin{center}
$u^{\mbox{\scriptsize SEFP}}_{ik}(\pi_j)=\displaystyle\sum_{\substack{o_h\in\pi_{ji}\\ u_{kh}>0}} u_{ih}$ \hspace{1cm} $\overline{u}^{\mbox{\scriptsize SEFA}}_{ik}(\Pi_j)=\displaystyle\sum_{\substack{h=1 \\ u_{kh}>0}}^j p_{ih}(\Pi_j)\cdot u_{ih}$
\end{center}

We note ${u}^{\mbox{\scriptsize SEFP}}_{ik}(\pi_{j})\leq {u}_{ii}(\pi_{j})$ and $\overline{u}^{\mbox{\scriptsize SEFA}}_{ik}(\Pi_{j})\leq \overline{u}_{ii}(\Pi_{j})$. A mechanism is \emph{shared envy-free ex post (SEFP)} iff, for each instance with $m\in\mathbb{N}$ items and allocation $\pi_m\in \Pi_m$ returned by the mechanism with positive probability, $\forall i,k\in N:u^{\mbox{\scriptsize SEFP}}_{ik}(\pi_m)\geq u_{ik}(\pi_m)$. A mechanism is {\em shared envy-free ex ante (SEFA)} iff, for each instance of $m\in\mathbb{N}$ items, $\forall i,k\in N:\overline{u}^{\mbox{\scriptsize SEFA}}_{ik}(\Pi_m)\geq \overline{u}_{ik}(\Pi_m)$. Shared envy-freeness coincides with envy-freeness with non-zero cardinal utilities. For this reason, shared envy-freeness is only possible in expectation. 

\begin{theorem}\label{thm:five}
A non-wasteful mechanism for online fair division is shared envy-free ex ante iff it
is ex ante equivalent to the \mylike\ mechanism.
\end{theorem}

\begin{proof}
If a mechanism is ex ante equivalent to \mylike, then it is envy-free ex ante. Every pair of agents receive each of their commonly liked item with the same probability. The mechanism is, therefore, shared envy-free ex ante. If a mechanism is shared envy-free ex ante, then the proof resembles the one of Theorem~\ref{thm:three}. In the step case, we consider round $j$ and agents $i,k$ that like item $o_j$. WLOG, assume that the mechanism is not ex ante equivalent to \mylike\ for item $o_j$ and $p_{i}(\Pi_j)<p_{k}(\Pi_j)$. By the hypothesis, the mechanism is ex ante equivalent to \mylike\ up to round $(j-1)$. Hence, $\overline{u}_{ik}(\Pi_{j-1})=\overline{u}^{\mbox{\scriptsize SEFA}}_{ik}(\Pi_{j-1})$ and $\overline{u}_{ki}(\Pi_{j-1})=\overline{u}^{\mbox{\scriptsize SEFA}}_{ki}(\Pi_{j-1})$. As the mechanism is shared envy-free ex ante up to round $j$, $p_{i}(\Pi_j)=p_{k}(\Pi_j)$. This contradicts our assumption. \qed
\end{proof}

If we limit ourselves to 0/1 utilities, we say that a mechanism is \emph{bounded envy-free ex post with 1 (BEFP)} iff, for each instance of $m\in\mathbb{N}$ items and $\pi_m\in\Pi_m$ returned by the mechanism with positive probability, $\forall i,k\in N: u_{ii}(\pi_m)+1\geq u_{ik}(\pi_m)$. For example, the \myblike\ mechanism is bounded envy-free ex post with 1 \cite{aleksandrov2015ijcai}. In fact, we can immediately conclude the following partial characterization.

\begin{corollary}\label{cor:one}
With 0/1 cardinal utilities, a non-wasteful mechanism for online fair division is bounded envy-free ex post with 1 if it returns a subset of the allocations returned by the \myblike\ mechanism.
\end{corollary}

Benade {et al.} \cite{benade2018} showed that the random assignment of each next item (i.e.\ \mylike) is asymptotically optimal in the ex post sense, with a bound of the (maximum) envy that increases as the number of rounds increases. Unfortunately, this means that we cannot put any trivial bound on the envy ex post in general. 

To sum up, we can use the \mylike\ or \myorp\ mechanism if we want envy-freeness ex ante. With 0/1 utilities, we can bound the ex post envy between agents to at most one unit of utility with the \myblike\ mechanism which also happens to be envy-free ex ante in this domain \cite{aleksandrov2015ijcai}.

\section{Pareto Efficiency}\label{sec:pe}

We consider lastly Pareto efficiency supposing agents act sincerely. With 0/1 utilities, each mechanism is Pareto efficient as the sum of agents' utilities in each returned allocation is $m$. This is not true in general. We start with Pareto efficiency ex post. The \myosd, \myorp\ and \mymax\ mechanisms are all Pareto efficient ex post. We might hope that a given Pareto efficient ex post mechanism returns some of the allocations returned by these three mechanisms. However, this does not hold as they may return only some of the Pareto efficient allocations. We illustrate this in Example~\ref{exp:three}.

\begin{example}\label{exp:three}
Let us consider the fair division of items $o_1$ and $o_2$ to agents $1$ and $2$ with utilities as in the below table.

\begin{center}
\begin{tabular}{|c|cc|} \hline
& item $o_1$ & item $o_2$ \\ \hline
agent 1 & \num{1} & \num{4}  \\
agent 2 & \num{2} & \num{3} \\ \hline
\end{tabular}
\end{center}

The allocation that gives $o_1$ to $1$ and $o_2$ to $2$ is Pareto efficient ex post. None of  \myosd, \myorp\ or \mymax\ returns this allocation. Note that \mypar\ does return it. \qed
\end{example}

By Example~\ref{exp:three}, we conclude that we cannot characterize all Pareto efficient ex post mechanisms in terms of allocations returned by the \myosd, \myorp\ and \mymax\ mechanisms. However, we can use the \mypar\ mechanism for this purpose.

\begin{theorem}\label{thm:six}
The \mypar\ mechanism returns only and all Pareto efficient ex post allocations.
\end{theorem}

\begin{proof}
By definition, the mechanism returns only PE ex post allocations. For this reason, we next only show that it returns all such allocations. Consider such an allocation $\pi_m$. Assume $\pi_m$ is not returned by it. Run the mechanism and follow $\pi_m$ until the first round $j\in(1,m]$ when some agent $i\in N$ gets $o_j$ in $\pi_m$ but $i$ is not feasible for $o_j$ given the sub-allocation $\pi_{j-1}$ of $\pi_m$ of $o_1$ to $o_{j-1}$. Such a round exists as $\pi_m$ is not returned by the mechanism. Further, $\pi_{j-1}$ is Pareto efficient ex post for $o_1$ to $o_{j-1}$. Otherwise, the mechanism would not get to round $j$ by following $\pi_m$. Also, the allocation extending $\pi_{j-1}$ by allocating $o_j$ to $i$ is Pareto efficient ex post. Otherwise, this allocation can be Pareto improved for $o_1$ to $o_j$ and together with the allocations of $o_{j+1}$ to $o_m$ in $\pi_m$  can Pareto improve $\pi_m$. This contradicts the Pareto efficiency of $\pi_m$. Hence, the allocation extending $\pi_{j-1}$ is Pareto efficient ex post. By the definition of the mechanism, it then follows that $i$ is feasible for $o_j$ which contradicts our assumption. Hence, $\pi_m$ is returned by the mechanism with positive probability. \qed
\end{proof}

By Theorem~\ref{thm:six}, we conclude that a non-wasteful mechanism for online fair division is Pareto efficient ex post iff it returns a subset of the allocations of the \mypar\ mechanism. Such a mechanism may not be strategy-proof. However, we can characterize \emph{all} mechanisms that are strategy-proof and Pareto efficient ex post. 

\begin{theorem}\label{thm:seven}
A non-wasteful mechanism for online fair division is strategy-proof and Pareto efficient ex post iff it is ex post equivalent to a probability distribution of the \myosd\ mechanisms.
\end{theorem}

\begin{proof}
We start with the ``if'' direction. If a mechanism is ex post equivalent to a probability distribution of  \myosd{\sc \footnotesize s}, then it is strategy-proof and Pareto efficient ex post as each \myosd. We next prove the ``only if'' direction. Consider a strategy-proof and Pareto efficient ex post mechanism and assume that it is not ex post equivalent to any probability distribution of \myosd{\sc \footnotesize s}. Hence, there is an instance, an allocation and $j\in[1,m]$ such that the mechanism and \myosd\ with some priority ordering $\sigma$ agree on $o_1$ to $o_{j-1}$ but the mechanism and any such \myosd\ disagree on $o_j$. WLOG, let the mechanism give $o_j$ to $1$ and \myosd\ with $\sigma$ give $o_j$ to $2$ such that $2$ is immediately before $1$ in $\sigma$. Both agents like item $o_j$. We can show that there is $o_k$ with $k<j$ such that $1$ and $2$ like $o_k$, and that $o_k$ is allocated to agent $2$ with both mechanisms. By Theorem~\ref{thm:one}, with the mechanism, the probabilities of $2$ for $o_k$ and $1$ for $o_j$ do not change for any positive bids of these agents for these items. WLOG, let then $u_{1j}=1$, $u_{1k}=2$, $u_{2j}=2$, $u_{2k}=1$. Hence, the allocation that extends $\pi_{j-1}$ by allocating $o_j$ to agent $1$ is not Pareto efficient ex post.
\qed
\end{proof}

Let us next add the ex ante properties. There is an unbounded number of Pareto efficient ex post and envy-free ex ante (or Pareto efficient ex ante) mechanisms that are not strategy-proof. To see this, consider the mechanism for the instance in Example~\ref{exp:two}, that runs the \myorp\ (or \mymax) mechanism on each other instance. Nevertheless, by Theorems~\ref{thm:three} and~\ref{thm:seven}, the only strategy-proof such mechanism is the \myorp\ mechanism.

\begin{corollary}\label{cor:two}
A non-wasteful mechanism for online fair division is strategy-proof, Pa\-reto efficient ex post and envy-free ex ante iff it is ex post equivalent to the \myorp\ mechanism.
\end{corollary}

A mechanism that is Pareto efficient ex post might not be Pareto efficient ex ante. For example, the \myorp\ mechanism is Pareto efficient ex post but not ex ante. To see this, consider the instance in Example~\ref{exp:one}. The reverse direction may also not hold. That is, a mechanism that is Pareto efficient ex ante may not necessarily be Pareto efficient ex post. We show this in Example~\ref{exp:four}.

\begin{example}\label{exp:four}
Consider the mechanism that runs \mymax\ on each instance except on the instance from Example~\ref{exp:one}. In this instance, the mechanism works as follows: agent $1$ gets $o_1$ and $o_2$ with probabilities $1$ and $1-\epsilon$, and agent $2$ gets these items with probabilities 0 and $\epsilon$ where $\epsilon > 0$. With this mechanism, agent $1$ gets expected utility $3 - 2\epsilon$, whilst agent $2$ gets expected utility $\epsilon$. This outcome is Pareto efficient ex ante for any $\epsilon<1/2$. But, there is one returned allocation that gives $o_1$ to agent $1$ and $o_2$ to agent $2$. This outcome is not Pareto efficient ex post. \qed
\end{example}

It is easy to see that the mechanism in Example~\ref{exp:four} is not strategy-proof. Interestingly, we can give a complete characterization of mechanisms that are strategy-proof, Pareto efficient ex post and Pareto efficient ex ante.

\begin{theorem}\label{thm:eight}
A non-wasteful mechanism for online fair division is strategy-proof, Pa\-reto efficient ex post and ex ante iff it is ex post equivalent to the \myosd\ mechanism.
\end{theorem}

\begin{proof}
We show the ``if'' direction. The mechanism returns the same allocation as \myosd. Hence, it is strategy-proof, Pareto efficient ex post and Pareto efficient ex ante. We next show the ``only if'' direction. By Theorem~\ref{thm:seven}, the mechanism is a probability distribution of \myosd{\sc \footnotesize s}. Suppose that there are at least two different allocations which are the result of different \myosd{\sc \footnotesize s} in this distribution. WLOG, assume that agent $1$ have the highest priority with probability $p_1\in (0,1)$, agent $2$ with $p_2\in (0,1-p_1]$ and agent $k\in N\setminus\lbrace 1,2\rbrace$ with $p_k\in[0,1-p_1-p_2]$. Suppose that agent $i\in\lbrace 1,2\rbrace$ likes all items with 1 except $o_i$ which they like with $u$, and agent $k\in N\setminus\lbrace 1,2\rbrace$ likes items positively. The expected utility of agent $i\in\lbrace 1,2\rbrace$ is $p_i\cdot (n-1+u)$ and the one of agent $k\in N\setminus\lbrace 1,2\rbrace$ is $p_k$ multiplied by the sum of their utilities. Consider now another distribution of allocations, in which agent $i\in\lbrace 1,2\rbrace$ gets $p_i$ for each item they like with 1 except items $o_1$, $o_2$, $p_1+p_2$ for item $o_i$ and $0$ for $o\in\lbrace o_1,o_2\rbrace\setminus\lbrace o_i\rbrace$ whereas agent $k\in N\setminus\lbrace 1,2\rbrace$ gets $p_k$ for each item. This allocation Pareto improves the allocation of the mechanism for $u>\max\lbrace (p_1/p_2),(p_2/p_1)\rbrace$. Hence, the mechanism is not Pareto efficient ex ante. Therefore, $p_1$ and $p_2$ cannot be both positive and, for this reason, each mechanism in the distribution gives the highest priority to the same agent. We can inductively show this for each priority.\qed
\end{proof}

We next observe one last difference to the offline setting where stochastic Pareto efficiency and envy-freeness are always possible \cite{bogomolnaia2001}. In online fair division, \emph{no} mechanism (even wasteful) satisfies Pareto efficiency ex ante and envy-freeness ex ante unless we consider simple 0/1 utilities (e.g.\ the \myblike\ mechanism).

\begin{theorem}\label{thm:nine}
With general cardinal utilities, \emph{no} mechanism for online fair division is envy-free ex ante and Pareto efficient ex ante. 
\end{theorem}

\begin{proof}
Consider an envy-free ex ante mechanism and the instance with non-zero utilities in Example~\ref{exp:one}. By Theorem~\ref{thm:four}, to ensure envy-freeness ex ante for $o_1$, the mechanism should give it to each agent with $1/2$. By Theorem~\ref{thm:four}, to ensure envy-freeness for both $o_1$ and $o_2$, the mechanism then should give $o_2$ to each agent with $1/2$. The expected utility of each agent is $3/2$. This expected allocation is Pareto dominated by the allocation in which each agent gets the item they value with \num{2}. Hence, the mechanism is not Pareto efficient ex ante.\qed
\end{proof}

To sum up, we might use the \myorp\ or \mypar\ mechanism for Pareto efficiency ex post, or the \mymax\ or \myosd\ mechanism for Pareto efficiency ex ante. With 0/1 utilities, we may also use the \mylike\ or \myblike\ mechanism.

\section{Conclusions}\label{sec:con}

We summarize all results in Table~\ref{tab:results} and Figure~\ref{fig:results}. For completeness, we add some simple results for the case of identical utilities when the \mypar\ and \mymax\ mechanisms become ex post equivalent to the \mylike\ mechanism, the \myblike\ mechanism becomes ex ante equivalent to the \mylike\ mechanism, and each of these becomes Pareto efficient as the sum of agents' utilities is a constant in each allocation.

\vspace{-0.5cm}
\begin{table}[h]
\centering
\caption{Axiomatic results. Key: $\star$ - the result follows from [Aleksandrov \emph{et al}., 2015].}

\resizebox{0.75\textwidth}{!}{
\begin{tabular}{|c|C|C|C|C|C|C|C|C|C|}
\hline
\multirow{2}{*}{\Large {\bf mechanism}} & \Large SP & \Large OSP & \Large EFA & \Large SEFA & \Large EFP & \Large SEFP & \Large BEFP & \Large PEA & \Large PEP \\ \cline{2-10}
& \multicolumn{9}{c|}{\Large general cardinal utilities} \\ \hline

\Large {\sc Online RP} & \Large $\checkmark$ & \Large $\checkmark$ & \Large $\checkmark$ & \Large $\checkmark$ & \Large $\times$ & \Large $\times$ & \Large $\times$ & \Large $\times$ & \Large $\checkmark$ \\ \hline
\Large {\sc Online SD} & \Large $\checkmark$ & \Large $\checkmark$ & \Large $\times$ & \Large $\times$ & \Large $\times$ & \Large $\times$ & \Large $\times$ & \Large $\checkmark$ & \Large $\checkmark$ \\ \hline
\Large {\sc Maximum Like} & \Large $\times$ & \Large $\times$ & \Large $\times$ & \Large $\times$ & \Large $\times$ & \Large $\times$ & \Large $\times$ & \Large $\checkmark$ & \Large $\checkmark$ \\ \hline
\Large {\sc Pareto Like} & \Large $\times$ & \Large $\times$ & \Large $\times$ & \Large $\times$ & \Large $\times$ & \Large $\times$ & \Large $\times$ & \Large $\times$ & \Large $\checkmark$ \\ \hline
\Large {\sc Like} & \Large $\checkmark^{\star}$ &  \Large $\checkmark$ & \Large $\checkmark^{\star}$ & \Large $\checkmark$ & \Large $\times^{\star}$ & \Large $\times$ & \Large $\times^{\star}$ & \Large $\times$ & \Large $\times$ \\ \hline
\Large {\sc Balanced Like} & \Large $\times^{\star}$ & \Large $\checkmark$ & \Large $\times^{\star}$ & \Large $\times$ & \Large $\times^{\star}$ & \Large $\times$ & \Large $\times^{\star}$ & \Large $\times$ & \Large $\times$ \\ \hline

& \multicolumn{9}{c|}{\Large identical cardinal utilities} \\ \hline
\Large {\sc Like} & \Large $\checkmark^{\star}$ & \Large $\checkmark$ & \Large $\checkmark^{\star}$ & \Large $\checkmark$ & \Large $\times^{\star}$ & \Large $\times$ & \Large $\times$ & \Large $\checkmark$ & \Large $\checkmark$ \\ \hline
\Large {\sc Balanced Like} & \Large $\times$ & \Large $\checkmark$ & \Large $\checkmark$ & \Large $\checkmark$ & \Large $\times^{\star}$ & \Large $\times$ & \Large $\times$ & \Large $\checkmark$ & \Large $\checkmark$ \\ \hline

& \multicolumn{9}{c|}{\Large binary cardinal utilities} \\ \hline
\Large {\sc Like} & \Large $\checkmark^{\star}$ & \Large $\checkmark$ & \Large $\checkmark^{\star}$ & \Large $\checkmark$ & \Large $\times^{\star}$ & \Large $\times$ & \Large $\times^{\star}$ & \Large $\checkmark$ & \Large $\checkmark$ \\ \hline
\Large {\sc Balanced Like} & \Large $\times^{\star}$ & \Large $\checkmark$ & \Large $\checkmark^{\star}$ & \Large $\times$ & \Large $\times^{\star}$ & \Large $\times$ & \Large $\checkmark^{\star}$ & \Large $\checkmark$ & \Large $\checkmark$ \\ \hline
\end{tabular}
}
\label{tab:results}
\end{table}
\vspace{-1.45cm}
\begin{figure}[h]
\centering
\resizebox{0.75\columnwidth}{4.5cm}{
\includegraphics{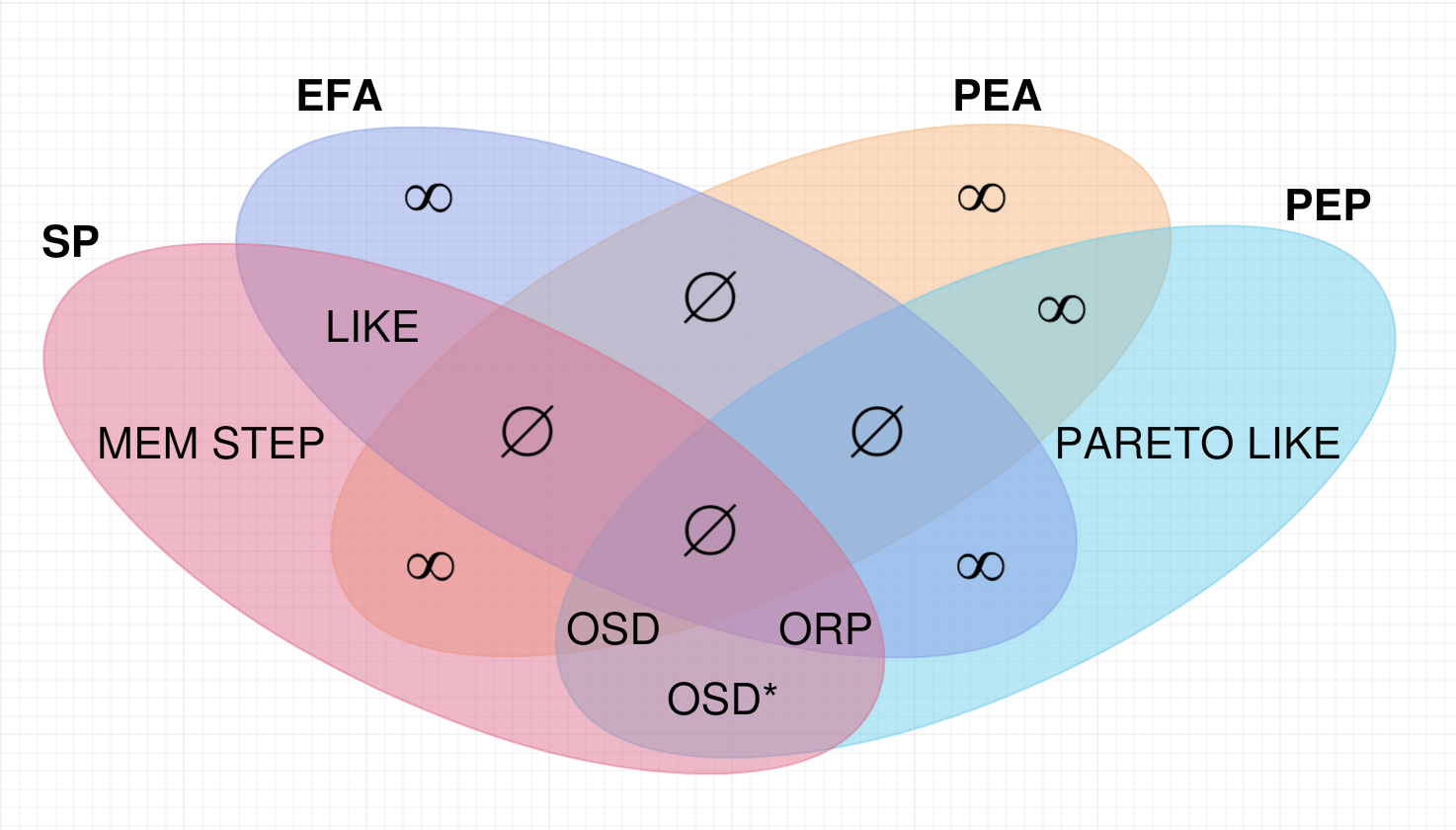}
}
\caption{General characterization results. Key: $\emptyset$ - no mechanisms, $\infty$ - inf. many mechanisms.}
\label{fig:results}
\end{figure}
\vspace{-0.5cm}

In future work, we will add quotas to our setting. And, we will extend our results to approximations of envy-freeness and general monotone utilities. 

\bibliographystyle{splncs}
\bibliography{online}

\end{document}